\documentclass[11pt, journal, letterpaper, oneside, onecolumn]{IEEEtran}

\newcommand{\linegap}{1}

\usepackage{theorem}
\usepackage{amsfonts}
\usepackage{amsmath}
\usepackage{amssymb}
\usepackage{setspace}
\usepackage{newlfont}
\usepackage{subfigure}
\usepackage[dvips]{graphicx}
\usepackage{calrsfs}

{\theoremstyle{plain}

}
{\theoremstyle{break} \theorembodyfont{\it}

 }

\newtheorem{prop}{Proposition}

\begin{document}

\title {Exact Analysis of Rate Adaptation Algorithms in Wireless LANs}

\author{Angad Singh and David Starobinski \\
Department of Electrical and Computer Engineering, Boston University \\
Email: {\it \{angad,staro\}@bu.edu}}

\maketitle

\thispagestyle{empty}

\begin{abstract}
Rate adaptation plays a key role in determining the performance of
wireless LANs. In this paper, we introduce a semi-Markovian
framework to analyze the performance of two of the most popular
rate adaptation algorithms used in wireless LANs, namely Automatic
Rate Fallback (ARF) and Adaptive Automatic Rate Fallback (AARF).
Given our modeling assumptions, the analysis is exact and provides
closed form expressions for the achievable throughput of ARF and
AARF. We illustrate the benefit of our analysis by numerically
comparing the throughput performance of ARF and AARF in two
different channel regimes. The results show that neither of these
algorithms consistently outperforms the other. We thus propose and
analyze a new variant to AARF, called Persistent AARF (or PAARF),
and show that it achieves a good compromise between the two
algorithms, often performing close to the best algorithm in each
of the studied regimes. The numerical results also shed light into the impact of MAC overhead on the performance of the three algorithms. In particular, they show that the more conservative strategies AARF and PAARF scale better as the bit rate increases.

\end{abstract}

\newpage

\setcounter{page}{1} \pagenumbering{arabic}

\section {Introduction}

Wireless LANs play a prominent role among wireless communication
systems~\cite{haastelecomm,ieee:80211,masked}. Most
wireless LANs support data transmission at multiple bit rates by
employing different modulation and channel encoding schemes. The
IEEE 802.11 WLAN family of standards is amongst the most popular
WLAN systems supporting data transmission at multiple
bit rates~\cite{ieee:80211,staro}. For instance, the IEEE 802.11b
standard allows transmissions at four different bit rates, i.e.,
1, 2, 5.5, and 11 Mbs, while the newer IEEE 802.11g standard
allows transmissions at 12 different bit rates ranging from 1 Mbs
to 54 Mbs.

The volatile nature of the wireless channel resulting from fading,
attenuation, and  interference from other radiation sources, makes
the task of rate selection in multi-rate WLANs system a key
feature for throughput optimization. A well designed algorithm
ought to select a bit rate for data transmission that maximizes
the instantaneous throughput. A key challenge however is that
channel quality usually fluctuates and, thus, any rate selection
algorithm must adapt to variations in the channel and network
conditions.

For IEEE 802.11 WLAN systems, several rate adaptation algorithms
have been proposed, see,
e.g.,~\cite{OrgARF,AARF,AARF_3,RRAA,CARA,Bicket}. Most of these
algorithms are rooted in the same design philosophy. They employ
open-loop rate adaptation schemes, run locally on the network
nodes, that dynamically determine the data transmission rate based
on certain statistics collected by the transmitting node. Two of
the most popular rate adaptation algorithms belonging to this
category are the Automatic Rate Fallback (ARF)~\cite{OrgARF} and
Adaptive Automatic Rate Fallback (AARF)~\cite{AARF} algorithms
that use consecutive successful or failed packet transmissions to
guide rate adaptation (cf. Section~\ref{sec:Related Work}).

In this paper, we propose a new analytical framework to evaluate
the performance of the ARF and AARF algorithms in wireless LANs
with random channels. Our analysis, based on the theory of
semi-Markov processes~\cite{ROSS}, is \emph{exact} and provides
\emph{closed-form} expressions for the throughput achieved by ARF
and AARF.   While our analysis necessarily relies on some
simplifying assumptions for the sake of tractability, it has the
clear advantage of providing meaningful insight into the impact of
various algorithm and channel parameters on the performance of
these algorithms.

To illustrate the benefits of our analysis, we present numerical
results comparing the performance ARF and AARF for different
channel regimes. Our numerical results clearly identify channel
regimes where AARF outperforms ARF, and are in line with
simulation and experimental results reported
in~\cite{Bicket,AARF}. More surprisingly, they also show that
there exist some practical regimes where ARF significantly
outperforms AARF.

Based on the insights gathered from our numerical analysis, we
propose a new variant to AARF, called Persistent Adaptive
Automatic Rate Fallback (PAARF). We show that the analysis of AARF
can easily be extended to that of PAARF. Numerical results show
that PAARF reaches a good compromise between ARF and AARF and
often gets very close to the best performing algorithm  in each of
the studied regimes.

It should be emphasized that the main goal of this paper is to provide general, qualitative insight
into the performance of rate adaption algorithms in wireless LANs, rather than conducting
detailed numerical modeling of a specific  protocol.

The rest of this paper is organized as follows. We discuss related
work  in Section~\ref{sec:Related Work} and introduce our model
and notations in Section~\ref{sec:Model and Notations}. We conduct
the analysis of ARF and AARF in Section~\ref{sec:analysis} and show how it can be generalized to handle MAC overhead. In Section~\ref{sec:Numercial Analysis}, we numerically compare the
performance of ARF and AARF and introduce the new PAARF algorithm.
We provide concluding remarks in Section~\ref{sec:Conclusions}.

\section {Related Work}
\label{sec:Related Work}

We first provide detail on the ARF and AARF algorithm and then
briefly discuss other relevant work. ARF~\cite{OrgARF} is the first documented rate adaptation
algorithm developed to optimize throughput performance in wireless
LAN devices.
ARF keeps transmitting at a given bit rate until a certain number of
\emph{consecutive} packets transmissions have either succeeded or
failed. Specifically, if $f$ consecutive packet transmissions fail
to get acknowledged at the current bit rate, then the next lower
bit rate (if there is such one) is selected for data transmission.
Similarly, if $s$ consecutive packet transmissions are acknowledged
without any re-transmissions at the current bit rate, then the next
higher bit rate (if there is such one) is selected for data
transmission. The default value of ARF parameters are $f=2$ and
$s=10$. ARF requires the maintenance of very little state
information.
Its simplicity has made it one of the most widely implemented
open-loop rate adaptation schemes in commercial 802.11 WLAN
devices~\cite{CARA}.

AARF~\cite{AARF} is derived from ARF. It tries to improve
throughput performance in scenarios where the packet success
probability at a certain bit rate is much higher than at the next
higher bit rate. The problem with ARF in such cases is that after
$s$ consecutive successful packet transmission at the low bit rate
it always attempts transmissions at the higher bit rate.
Instead, AARF implements a binary exponential back-off procedure
whereby after every failed probe packet transmission at the higher
bit rate, AARF doubles (up to some maximum value) the threshold
number of consecutive packet transmissions required at the current
bit rate before attempting a packet transmission at the next
higher bit rate. Thus, AARF initially looks for $s$ consecutive
successful packet transmissions at the current bit rate after
which it sends a probe packet at the next higher bit rate. If the
probe packet transmission is successful, then AARF switches
 to the higher bit rate. Otherwise, it stays in the
current bit rate. In that case, the next probe packet transmission
at the higher bit rate is attempted only after $2s$ consecutive
successful packet transmissions at the current bit rate and so
forth.

Several papers have proposed various modifications and
improvements over the basic ARF and AARF rate adaptation
algorithms, see e.g.,~\cite{Bicket,RRAA,CARA,CARA11,SNR1}. Several
of these algorithms, e.g.,~\cite{CARA,CARA11,SNR1},  are based on
variations of ARF, and thus we expect our analytical models to be
useful to evaluate their performance as well.

So far, most of the evaluation of rate adaptation algorithms has
been carried out through
simulations~\cite{AARF,CARA,AARF_3,AARF_7,AARF_12,SNR1,EACK} or
experiments on actual testbed networks~~\cite{Bicket,RRAA}.
 Although there exists analytical work for multi-rate wireless networks, see
e.g.,~\cite{WLANsingle1,WLANsingle2}, that work assumes that each
node always transmits at a fixed rate.
An exception is the recent work of~\cite{Infocomm}, where, among
other contributions, the authors present a Markov chain model of
ARF. The present work differs from~\cite{Infocomm} in several
aspects. First, we use the more general theory of semi-Markov
processes to analyze rate adaptation algorithms. Thus, the
analysis provides means to evaluate the effects of MAC overhead (e.g., binary exponential back-off),
which can be significant at high bit rate rates.
Second, we also provide an analysis of the AARF algorithm and
numerical comparisons between the performance of ARF and AARF.
Finally, we introduce the new PAARF algorithm and compare its
performance to the two other algorithms. A preliminary version of this paper appeared in~\cite{SECON07}. The present work primarily adds to that earlier work by modeling and analyzing the impact of MAC overhead (cf.~Sections~\ref{sec:ARF_IEEE} and~\ref{sec:impact}).

\section {Model and Notations}
\label{sec:Model and Notations}

Our goal in this paper is to conduct an exact analysis of the
performance of the ARF and AARF algorithms in wireless LANs, such
as IEEE 802.11 networks. As such, a certain number of assumptions
are necessary in order to keep the analysis tractable.

In order to decouple the behavior of the above algorithms from
other MAC and higher-layer mechanisms, we focus our attention on
the behavior of ARF and AARF for a single source-destination pair
(e.g., a mobile node and a base station). We note that most
wireless LANs operate at low load and, thus, it is typical that,
at any given point of time,  only one pair of nodes
communicates~\cite{Traffic}. We assume that the source is greedy,
i.e., it has always packets to transmit.

The source can transmit packets at $N$ different bit rates,
denoted by $R_1, R_2, \ldots, R_N$ in units of bit/s. Without any
limitations of generality, we  assume that these rates are sorted
from the lowest to the highest, i.e.,  $R_{1}$ represents the
lowest available bit rate and $R_{N}$ the highest. At each
bit rate $R_{i}$, we denote the probability of a successful packet
transmission by $\alpha_i$, where $0<\alpha_{i}<1$. This
probability is assumed to be independent of any other events.

We use the random variable $\ell$ to represent the length (in
bits) of a packet. This variable follows an arbitrary i.i.d.
distribution (i.e., not necessarily exponential). The mean packet
length is denoted by $\overline{\ell}$.

Next, define $f_{i}$ to be  the long run proportion of time during
which packet transmission is carried out at the bit rate $R_{i}$.
We can then express the steady-state throughput $\tau$ as follows:
\begin{equation}\label{eq:throughput}
\tau=\sum_{i=1}^{N}{f_{i}\alpha_{i}R_{i}}.
\end{equation}
The key for characterizing the throughput performance of ARF and
AARF resides in deriving an expression for $f_{i}$ for each of the
algorithms.


\section{Analysis}
\label{sec:analysis}

In Sections~\ref{sec:ARF} and~\ref{sec:AARF},  we analyze the throughput performance of ARF and AARF. The analysis will lead to closed-form expressions for the
throughput delivered by these algorithms. The focus of the
analysis presented in these sections is to understand the basic
behavior of these algorithms independent of the specific underlying
MAC protocol. Hence, initially, we will not deal with protocol-specific
details, such as back-off retransmissions, control packets, and inter-frame spacings in the analysis. Those will be considered in Section~\ref{sec:ARF_IEEE}.

\subsection{ARF}
\label{sec:ARF}

Based on our statistical assumptions, we next show that the behavior
of ARF can be analyzed using the theory of semi-Markov
processes~\cite{ROSS}. Similar to a Markov process, a semi-Markov
process transitions between different states. Upon \emph{entering} a
certain state, the time spent in that state and the transition
probabilities to the various possible next states depend only on the
present state and are independent of the history. However, contrary
to a standard Markov process, the time spent in each state follows a
general distribution, which is not necessarily memoryless. Thus, a
semi-Markov process is not Markovian at an arbitrary point of time.
However, one can create an embedded Markov chain by sampling the
original process at moments of transition to a new state.

Now, define state $i$ to be the state in which packets are
transmitted at the bit rate $R_i$. Clearly, upon entering state $2
\leq i \leq N-1$, the time spent in that state and the transition
probabilities to state $i-1$ and $i+1$ depend only on the
parameters $\alpha_i$, $R_i$, $\overline{\ell}$, $s$,  and $f$
and, thus, are independent of the past.
Similar arguments apply for the time spent in states 1 and $N$.
Thus, the  behavior of ARF can be modeled using a semi-Markov
process. The embedded Markov chain for the problem at hand is
depicted in Fig.~\ref{fig:ARF-Macro}. The quantities $p_{i,j}$
shown in the figure represent the transition probabilities from
state $i$ to state $j$.

\begin{figure}[t]
\centering
\resizebox{3in}{!}{\includegraphics{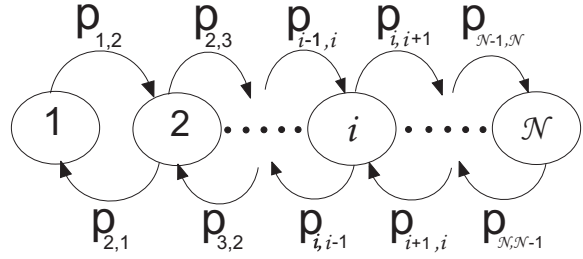}}
\caption{Embedded Markov chain modeling  ARF behavior at the
moments of transitions to a new state. State $i$ represents packet
transmissions at rate $R_i$.} \label{contblock}
\label{fig:ARF-Macro}
\end{figure}\renewcommand{\baselinestretch}{\linegap}

As per Eq.~(\ref{eq:throughput}), in order to find an expression
for the throughput of ARF we need to calculate $f_{i}$, i.e., the
long run proportion of time data transmission is carried out at
the bit rate $R_{i}$. Let $p_{i}$ represent the steady-state
probability of finding the semi-Markov process in state $i$. From
the definition of state $i$, we immediately see that
$f_{i}=p_{i}$.

In order to compute $p_i$, we will exploit the mathematical
properties of semi-Markov processes~\cite{ROSS}. Specifically,
define $\mu_i$ to be the mean time spent in each state $i$ of the
semi-Markov process and $\pi_i$ to be the steady-state proportion
of transitions into state $i$. The latter also corresponds to the
steady-state fraction of time the embedded Markov chain associated
with the process finds itself in state $i$. Then, it can be
shown~\cite{ROSS}:
\begin{equation}\label{eq:frac}
p_{i}=\frac{\pi_{i}\mu_{i}}{\sum_{i=1}^{N}{\pi_{i}\mu_{i}}}.
\end{equation}

The embedded Markov chain shown in Fig.~\ref{fig:ARF-Macro} is a
simple birth-death process~\cite{ROSS}. Thus, the steady-state
probabilities $\pi_i$, for each state $i \geq 2$, can readily be
expressed as follows:
\begin{equation}\label{eq:pi}
\pi_{i}= \pi_{1}\prod_{k=1}^{i-1}{\frac{p_{k,k+1}}{p_{k+1,k}}}.
\end{equation}
In order to calculate $\pi_{1}$, we apply the normalization
condition and get
\begin{equation}\label{eq:pi_1}
\pi_{1}=\frac{1}{1+\sum_{i=2}^{N}{\prod_{k=1}^{i-1}\frac{p_{k,k+1}}{p_{k+1,k}}}}.
\end{equation}

In order to complete the analysis, it just remains to derive
expression for the average time spent in each state $\mu_{i}$ and
the transition probabilities $p_{i,j}$.
Toward this end, we need to model the operations of ARF within
each state $i$, corresponding to transmissions at bit rate
$R_{i}$.
Specifically, we need to keep track of the number of consecutive
successful or failed packet transmissions. The state diagram shown
in Fig.~\ref{fig:ARF-micro} models the behavior of ARF at a given
bit rate $R_{i}$, where $1<i<N$. The initial state is state
$\cal{S}_0^i$. Each subsequent successful packet transmission
leads to a transition into some state $\cal{S}_j^i$, where $j$
represents the number of consecutive successful packet
transmissions. Similarly each failed packet transmission leads to
a transition into some state $\cal{S}_{-j}^i$ where $j$ represents
the number of consecutive failed packet transmissions. States
$\cal{S}_s^i$ and $\cal{S}_{-f}^i$ are termination states after
which packet transmissions will occur at bit rates $R_{i+1}$ and
$R_{i-1}$, respectively. The state diagrams for bit rates $R_1$
and $R_N$ are similar, except that there is no need to account for
consecutive failed packet transmissions  and consecutive
successful packet transmissions, respectively.

\begin{figure}[t]
\centering
\resizebox{4in}{!}{\includegraphics{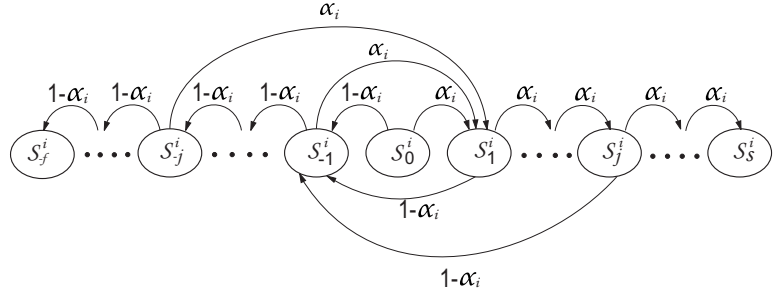}}
\caption{ARF operation at an intermediate bit rate $R_{i}$. States
$\cal{S}_j^i$ and $\cal{S}_{-j}^i$ represent respectively $j$
consecutive successful and failed packet transmissions.}
\label{fig:ARF-micro}
\end{figure}\renewcommand{\baselinestretch}{\linegap}

Now, let the random variable $X_{i}(j)$ represent the number of
packet transmissions at bit rate $R_i$  before reaching state
$\cal{S}_s^i$ or state $\cal{S}_{-f}^i$, starting from state
$\cal{S}_j^i$. The quantity $\overline{X_{i}}(0)$ represents the
average number of packet transmission starting from state
$\cal{S}_0^i$ until one of the termination states is reached. One
can express $\mu_i$ as a function of $\overline{X_{i}}(0)$ in the
following way:
\begin{equation}\label{eq:mui}
\mu_{i}=\overline{X_{i}}(0)\frac{\overline{\ell}}{R_{i}}.
\end{equation}

The special structure of the state diagram shown in
Fig.~\ref{fig:ARF-micro} allows to provide a closed-form
expression for $\overline{X_{i}}(0)$, as given by the following
proposition.

\begin{prop}
\label{prop1}
 Let $\overline{X_{i}}(0)$ represents the expected
number of packet transmission in state $i$ of
Fig.~\ref{fig:ARF-Macro}. Then, the following holds:

\begin{eqnarray}\label{eq:prop1}
\overline{X_{i}}(0)=\left\{\begin{array}{ll}
\frac{\sum_{j=0}^{s-1}{(\alpha_{i})^j}}{(\alpha_{i})^s},
& \qquad \mbox{for $i=1$;}\\[0.15in]
\frac{\sum_{j=0}^{s-1}{(\alpha_{i})^j}{\sum_{j=0}^{f-1}{{(1-\alpha_{i})^j}}}}
{1-\sum_{j=1}^{s-1}{(\alpha_{i})^j}\sum_{j=1}^{f-1}{{(1-\alpha_{i})^j}}},
& \qquad \mbox{for $1<i<N$ with $s>1$ and $f>1$;}\\[0.15in]
\sum_{j=0}^{s-1}{(\alpha_{i})^j}{\sum_{j=0}^{f-1}{{(1-\alpha_{i})^j}}},
& \qquad \mbox{for $1<i<N$ with $s=1$ or $f=1$;}\\[0.15in]
\frac{\sum_{j=0}^{f-1}{{(1-\alpha_{i})^j}}}{{(1-\alpha_{i})^f}},
& \qquad \mbox{for $i=N$}.
\end{array}
\right.
\end{eqnarray}
\end{prop}

\begin{proof}

We will prove the proposition for the case $1<i<N$ with $s>1$ and
$f>1$. The proof for the other cases is similar.

The proof follows a two step approach. The first step is to show
that the following two equations hold:
\begin{equation}
\label{eq:proof1}
\overline{X_{i}}(s-u)=
\sum_{k=1}^{u}{(\alpha_{i})^k}+(1-\alpha_{i})(1+\overline{X_{i}}(-1))\sum_{k=0}^{u-1}{(\alpha_{i})^k}, 
\qquad \mbox{for $0<u<s$};
\end{equation}
\begin{equation}
\label{eq:proof2}
\overline{X_{i}}(-f+v)=
\sum_{k=1}^{v}{(1-\alpha_{i})^k}+\alpha_{i}(1+\overline{X_{i}}(1))\sum_{k=0}^{v-1}{(1-\alpha_{i})^k}, \qquad \mbox{for $0<v<f$}.
\end{equation}

We will prove Eq.~(\ref{eq:proof1}) using mathematical induction.
The proof of Eq.~(\ref{eq:proof2}) is conducted in a similar
manner.

First, we prove the basis of the induction, i.e., we consider the
case $u=1$. Consider the average number of transmissions starting
from state $\cal{S}_{s-1}^i$. With probability $\alpha_i$, the
next packet transmission is successful and ARF exits the current
bit rate to the next higher bit rate. Otherwise, with probability
$1-\alpha_i$, the transmission fails and the process moves to
state $\cal{S}_{-1}^i$. We thus have
\begin{equation}\label{eq:induction1}
\overline{X_{i}}(s-1)={\alpha_{i}}\cdot1+(1-\alpha_{i})(1+\overline{X_{i}}(-1)).
\end{equation}
This equation is equivalent to Eq.~(\ref{eq:proof1}) for $u=1$ and
thus proves the basis of the induction.



Next, we prove the induction step. Assume Eq.~(\ref{eq:proof1})
holds true for $u=m$, where $1 \leq m<s-1$, that is,
\begin{equation}\label{eq:induction2}
\overline{X_{i}}(s-m)=\sum_{k=1}^{m}{(\alpha_{i})^k}+(1-\alpha_{i})(1+\overline{X_{i}}(-1))\sum_{k=0}^{m-1}{(\alpha_{i})^k}.
\end{equation}

Now assume that the process is in state $\cal{S}_{s-(m+1)}^i$.
With probability $\alpha_i$, the next transmission is successful
and the process moves to state $\cal{S}_{s-m}^i$. Otherwise, with
probability $1-\alpha_i$, the transmission fails and the process
moves to state $\cal{S}_{-1}^i$. Thus,
\begin{equation}
\overline{X_{i}}(s-(m+1))=
\alpha_{i}(1+\overline{X_{i}}(s-m))+(1-\alpha_{i})(1+\overline{X_{i}}(-1)).
\label{eq:induction4}
\end{equation}

Substituting Eq.~(\ref{eq:induction2}) into
Eq.~(\ref{eq:induction4}), we obtain
\begin{equation}
\overline{X_{i}}(s-(m+1))=
\sum_{k=1}^{m+1}{(\alpha_{i})^k}+(1-\alpha_{i})(1+\overline{X_{i}}(-1))\sum_{k=0}^{m}{(\alpha_{i})^k},
\label{eq:induction3}
\end{equation}
hence proving the induction step.

Now we proceed with the second step of the proof. We note that
after the process enters state $\cal{S}_0^i$, it either moves to
state $\cal{S}_1^i$ (with probability $\alpha_i$) or to state
$\cal{S}_{-1}^i$ (with probability $1-\alpha_i$). Therefore,
\begin{equation}\label{eq:proof5}
\overline{X_{i}}(0)=\alpha_{i}(\overline{X_{i}}(1)+1)+(1-\alpha_{i})(\overline{X_{i}}(-1)+1).
\end{equation}

Substituting $u=s-1$ in Eq.~(\ref{eq:proof1}) and $v=f-1$ in
Eq.~(\ref{eq:proof2}), we get
\begin{equation}
\overline{X_{i}}(1)=
\sum_{k=1}^{s-1}{(\alpha_{i})^k}+(1-\alpha_{i})(1+\overline{X_{i}}(-1))\sum_{k=0}^{s-2}{(\alpha_{i})^k};
\label{eq:proof3}
\end{equation}
\begin{equation}
\overline{X_{i}}(-1)=
\sum_{k=1}^{f-1}{(1-\alpha_{i})^k}+\alpha_{i}(1+\overline{X_{i}}(1))\sum_{k=0}^{f-2}{(1-\alpha_{i})^k}.
 \label{eq:proof4}
\end{equation}
Equations~(\ref{eq:proof5}),~(\ref{eq:proof3}) and~(\ref{eq:proof4})
provide three linear equations in three unknowns (i.e,
$\overline{X_{i}}(0),\overline{X_{i}}(-1)$ and
$\overline{X_{i}}(1))$ from which obtain the expression of
$\overline{X_{i}}(0)$ given by Proposition 1 for the case $1<i<N$
with $s>1$ and $f>1$.
\end{proof}



The next proposition provides expressions for the transition
probabilities  of the embedded Markov chain shown in
Fig.~\ref{fig:ARF-Macro},  for $1<i<N$. 
To prove this proposition, we compute the probability of getting
from state $\cal{S}_0^i$ to state $\cal{S}_s^i$ , which
corresponds exactly to $p_{i,i+1}$.

\begin{prop}
\label{prop2}
 Let $p_{i,i+1}$ be the transition probability of
switching from state $i$ to state $i+1$. Then,

\begin{eqnarray}\label{eq:prop2} p_{i,i+1}=\left\{\begin{array}{ll}
\frac{(\alpha_{i})^s{\sum_{j=0}^{f-1}{{(1-\alpha_{i})^j}}}}{1-[\sum_{j=1}^{s-1}{{(\alpha_{i})}^j}\sum_{j=1}^{f-1}{{(1-\alpha_{i})^j}}]}
& \qquad \mbox{for $1<i<N$ with $s>1$ and $f>1$;}\\[0.15in]
(\alpha_{i})^s{\sum_{j=0}^{f-1}{{(1-\alpha_{i})^j}}} & \qquad\mbox{for
$1<i<N$ with $s=1$ or $f=1$}.
\end{array}
\right.
\end{eqnarray}
In addition, we have $p_{1,2}=p_{N,N-1}=1$, and
$p_{i,i-1}=1-p_{i,i+1}$ for $1<i<N$.
\end{prop}
\begin{proof}

Define $q_{i}(j)$ to be the probability of reaching state
$\cal{S}_s^i$ from state $\cal{S}_j^i$. Therefore,
$p_{i,i+1}=q_{i}(0)$.

We outline the proof of the proposition for the case $1<i<N$ with
$s>1$ and $f>1$. Similar to Proposition~\ref{prop1}, the proof
follows a two step approach. The first step is to prove that the
following two equations hold, which can be done via induction as
in the proof of Proposition~\ref{prop1}:
\begin{equation}
q_{i}(s-u)=
\alpha_{i}^u+q_{i}(-1)(1-\alpha_{i})\sum_{k=0}^{u-1}{{\alpha_{i}}^k}\mbox{
  for $0<u<s$};
\label{eq:prop2proof1}
\end{equation}
\begin{equation}
q_{i}(-(f-v))=
q_{i}(1)\alpha_{i}\sum_{k=0}^{v-1}{(1-\alpha_{i})^k}\mbox{ for
$0<v<f$}.
\label{eq:prop2proof2}
\end{equation}

Next, we note that
\begin{equation}\label{eq:prop2proof5}
q_{i}(0)=\alpha_{i}q_{i}(1)+(1-\alpha_{i})q_{i}(-1),
\end{equation}
and substituting $u=s-1$ in Eq.~(\ref{eq:prop2proof1}) and $v=f-1$
in Eq.~(\ref{eq:prop2proof2}) we have,
\begin{equation}\label{eq:prop2proof3}
q_{i}(1)={\alpha_{i}}^{s-1}+q_{i}(-1)(1-\alpha_{i})\sum_{k=0}^{s-2}{(\alpha_{i})^k};
\end{equation}
\begin{equation}\label{eq:prop2proof4}
q_{i}(-1)=q_{i}(1)\alpha_{i}\sum_{k=0}^{f-2}{(1-\alpha_{i})^k}.
\end{equation}




Solving Eqs.~(\ref{eq:prop2proof5}),~(\ref{eq:prop2proof3})
and~(\ref{eq:prop2proof4}) for $q_{i}(0)$, we obtain the
expression of $p_{i,i+1}$ given by Proposition~\ref{prop2} for the
case $1<i<N$ with $s>1$ and $f>1$.

\end{proof}

Using
Eqs.~(\ref{eq:frac}),~(\ref{eq:pi}),~(\ref{eq:pi_1}),~(\ref{eq:mui})
and Propositions~\ref{prop1} and~\ref{prop2}, we thus have derived
closed-form expressions for $p_{i}$, where $1 \leq i \leq N$, as a
function of the parameters $\alpha_i, R_i,\overline{\ell}, s, f$
and $N$. Since $f_i=p_i$, an expression for the throughput of ARF
follows immediately from Eq.~(\ref{eq:throughput}).

\subsection{AARF}
\label{sec:AARF}

The behavior of AARF is conceptually similar to that of ARF and
its analysis can also be carried out using a semi-Markov process
formulation. The complexity of the analysis lies in modeling the
back-off procedure of AARF, which requires properly defining the
states of the semi-Markov process.

To model the operation of AARF at each bit rate $R_{i}$, where $1
\leq i \leq N$, we define the ``fall back'' states $i_{\beta}$ and
the ``probe states'' $i_{\beta}^{+1}$, as illustrated in
Fig.~\ref{fig:AARF-macro}. The variable $\beta$, where $0 \leq
\beta \leq \beta_{max}$, is indicative of the current back-off
stage. Thus, if the process is in state $i_{\beta}$, there must be
${2^\beta} s$ consecutive successful packet transmissions before
the process moves to probe state $i_{\beta}^{+1}$, where a probe
packet is transmitted at rate $R_{i+1}$. If the probe packet is
successfully transmitted then the process transitions to state
${(i+1)}_{0}$. Otherwise, the process moves to the next fall back
state, i.e., state $i_{\beta+1}$. Similar to ARF, if the process
is in some state $i_{\beta}$ and experiences $f$ consecutive
packet transmission failures then it transitions to state
${(i-1)}_{0}$ (except for the case $i=1$, where the process
remains in the same state). The state $i_{\beta_{max}}$ represents
the maximum fall back state. The process keeps returning to that
state until the transmission of a probe packet at rate $R_{i+1}$
is successful or $f$ consequent packet failures occur. Finally, we
note that there are no fall back states at rate $R_N$, and thus
there is only one state $N_0$ which is defined the same way as
state $N$ in ARF.

\begin{figure}[t]
\centering
\resizebox{3in}{!}{\includegraphics{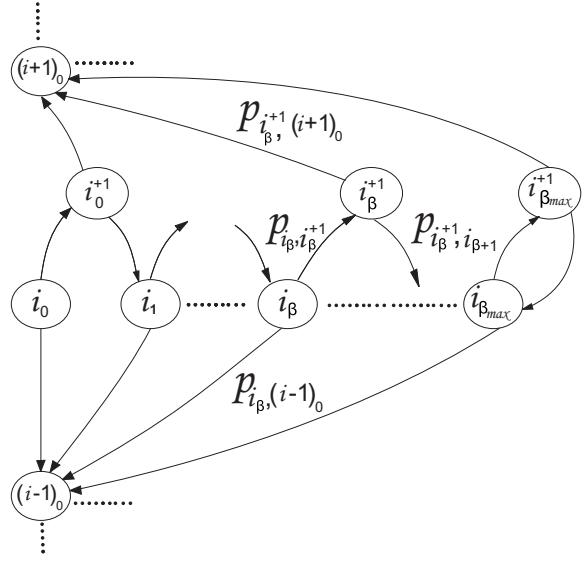}}
\caption{Embedded Markov chain modeling AARF behavior at the
moments of transitions to a new state. The variable $\beta$ is an
indicator of the back-off stage. States $i_\beta$ correspond to
``fall back'' states, in which transmissions take place at rate
$R_i$, and states $i^{+1}_\beta$ correspond to ``probe states'',
in which transmissions take place at rate $R_{i+1}$.}
\label{fig:AARF-macro}
\end{figure}\renewcommand{\baselinestretch}{\linegap}

Similar to ARF, whenever the process enters one of the above
defined states, the time spent in each state and the transition
probabilities to the next possible states are independent of the
history. Thus, the behavior of AARF can be modeled using a
semi-Markov process. The embedded Markov chain for this process is
shown in Fig.~\ref{fig:AARF-macro}.

As per Eq.~(\ref{eq:throughput}), in order to find an expression
for the throughput of AARF, we need to calculate $f_{i}$, i.e.,
the long run proportion of time data transmission is carried out
at the bit rate $R_{i}$. The quantities $f_{i}$ can be expressed
as a function of $p_{i_{\beta}}$ and $p_{(i-1)^{+1}_{\beta}}$
which are defined as the steady-state probabilities of finding the
semi-Markov process in either the fall-back state $i_{\beta}$ or
the probe packet state $(i-1)^{+1}_{\beta}$ respectively.
Specifically, we have
\begin{equation}\label{eq:f_AARF}
f_{i}=\sum_{\beta=0}^{\beta_{max}}(p_{i_{\beta}}+p_{(i-1)^{+1}_{\beta}})
\qquad \mbox{ for } 1 \leq i \leq N,
\end{equation}
where by definition $p_{{0}^{+1}_{\beta}}=0$, and
$p_{N_{\beta}}=0$ for $\beta \geq 1$.

As in the previous section, we can find expressions for
$p_{i_{\beta}}$ and $p_{i^{+1}_{\beta}}$ by computing i) the
average time spent in each state of the semi-Markov process; ii)
the transition probabilities of the embedded Markov chain; and
iii) the steady-state probabilities of the embedded Markov chain.

We start with items i) and ii). Consider first the probe states. The
average time spent in state~$i^{+1}_{\beta}$  is simply
$\mu_{i^{+1}_{\beta}}=\overline{\ell}/R_{i+1}$.
The transition probabilities out of the probe states, for $\beta < \beta_{max}$, are given by
$p_{(i^{+1}_{\beta},(i+1)_0)}  =
\alpha_{i+1}$ and
$p_{(i^{+1}_{\beta},i_{\beta+1})}  =  1-\alpha_{i+1}$.
For the case $\beta=\beta_{max}$, we have
$p_{(i^{+1}_{\beta_{max}},(i+1)_0)} = \alpha_{i+1}$ and
$p_{(i^{+1}_{\beta_{max}},i_{\beta_{max}})} =  1-\alpha_{i+1}$.

The behavior of AARF in the fall back states ${i_{\beta}}$ is very
similar to that of ARF in state $i$, except that the number of
consecutive successful transmissions required before transmitting
at the next higher bit rate is $b_{\beta}={2^\beta} s$ instead of
just $s$. Thus, we can obtain an expression for $\mu_{i_{\beta}}$ by simply replacing $s$ by $b_{\beta}$ in Proposition~\ref{prop1}.

The last item to complete the analysis is to compute the
steady-state probabilities of the embedded Markov chain
$\pi_{i_{\beta}}$ and $\pi_{i^{+1}_{\beta}}$.  Once this is done,
the proportion of time spent by AARF in each state is given by the
following expressions that are analogous to Eq.~(\ref{eq:frac}):
\begin{equation}\label{eq:pfallback_AARF}
p_{i_{\beta}}=\frac{\pi_{i_{\beta}}\mu_{i_{\beta}}}{\sum_{i=1}^{N}{\sum_{\beta=1}^{\beta_{max}}{(\pi_{i_{\beta}}+\pi_{(i-1)^{+1}_{\beta}})}}}; \qquad
p_{i^{+1}_{\beta}}=\frac{\pi_{i^{+1}_{\beta}}\mu_{i^{+1}_{\beta}}}{\sum_{i=1}^{N}{\sum_{\beta=1}^{\beta_{max}}{(\pi_{i_{\beta}}+\pi_{(i-1)^{+1}_{\beta}})}}},
\end{equation}
where by definition $\pi_{{0}^{+1}_{\beta}}=0$, and
$\pi_{N_{\beta}}=0$ for $\beta \geq 1$.

We next show that the seemingly complex structure of the embedded
Markov chain shown in Fig.~\ref{fig:AARF-macro} has the remarkable
property of collapsing into a simple birth-death process.

First, we observe that the steady probabilities of the states at
level $i$, namely $\pi_{i_{\beta}}$ and $\pi_{i^{+1}_{\beta}}$,
can all be expressed as a function of $\pi_{i_{0}}$. This is done
by taking contours around each state of level $i$ in order, that
is, $i^{+1}_{0}, i_{1}, i^{+1}_{1}, \ldots$, and writing the
balance equations for each. The expressions are as follows:

\begin{eqnarray}\label{eq:piprobepacket_AARF}
\pi_{i^{+1}_{\beta}}=\left\{\begin{array}{ll}
\pi_{i_{0}}p_{(i_{\beta},i^{+1}_{\beta})}\prod_{k=0}^{\beta-1}{[p_{(i_{k},i_{k}^{+1})}p_{(i^{+1}_{k},{i}_{k+1})}]},
& \qquad \mbox{for $0\leq\beta<\beta_{max}$ and $1\leq i<N$};\\[0.15in]
\frac{\pi_{i_{0}}p_{(i_{\beta},i^{+1}_{\beta})}\prod_{k=0}^{\beta-1}{[p_{(i_{k},i_{k}^{+1})}p_{(i^{+1}_{k},{i}_{k+1})}]}}{1-p_{(i^{+1}_{\beta},{i}_{\beta})}p_{(i_{\beta},{i}^{+1}_{\beta})}},
& \qquad \mbox{for $\beta=\beta_{max}$ and $1\leq i<N$};
\end{array}
\right.
\end{eqnarray}
\begin{eqnarray}\label{eq:pifallback_AARF}
\pi_{i_{\beta}}=\left\{\begin{array}{ll}
\pi_{i_{0}}\prod_{k=0}^{\beta-1}{[p_{(i_{k},i_{k}^{+1})}p_{(i^{+1}_{k},{i}_{k+1})}]},
& \qquad \mbox{for $0<\beta<\beta_{max}$ and $1\leq i<N$};\\[0.15in]
\frac{\pi_{i_{0}}\prod_{k=0}^{\beta-1}{[p_{(i_{k},i_{k}^{+1})}p_{(i^{+1}_{k},{i}_{k+1})}]}}{1-p_{(i^{+1}_{\beta},{i}_{\beta})}p_{(i_{\beta},{i}^{+1}_{\beta})}},
& \qquad \mbox{for $\beta=\beta_{max}$ and $1\leq i<N$}.
\end{array}
\right.
\end{eqnarray}
Now, at equilibrium, the rate of transitions from level $i$ to
level $i+1$ must be the same as that from level $i+1$ to level
$i$. Thus,
\begin{equation}
\sum_{\beta=0}^{\beta_{max}}
(p_{i^{+1}_{\beta},{{(i+1)}_{0}}})\pi_{i^{+1}_{\beta}}=\sum_{\beta=0}^{\beta_{max}}(p_{{(i+1)}_{\beta},{i_{0}}})\pi_{{(i+1)}_{\beta}},
\mbox{ for $1\leq i<N$}.\label{eq:birthdeath_AARF}
\end{equation}
Using Eq.~(\ref{eq:piprobepacket_AARF}), all the individual terms in
the lhs of Eq.~(\ref{eq:birthdeath_AARF}) can be expressed as a
function of $\pi_{i_{0}}$, while, using
Eq.~(\ref{eq:pifallback_AARF}), all the individual terms in the rhs
of Eq.~(\ref{eq:birthdeath_AARF}) can be expressed as a function of
$\pi_{(i+1)_{0}}$, leading to balance equations similar to a
birth-death process. Using Eq.~(\ref{eq:birthdeath_AARF}), we can
then express all the steady-state probabilities as a function of
$\pi_{1_{0}}$.
Finally, we can resort to the normalization condition  to evaluate
$\pi_{1_{0}}$, i.e.,
\begin{equation}\label{eq:nc_AARF}
\sum_{i=1}^{N}{\sum_{\beta=0}^{\beta_{max}}{(\pi_{i_{\beta}}+\pi_{(i-1)^{+1}_{\beta}})}}=1,
\end{equation}
and our analysis is complete.

\subsection{Accounting for MAC Overhead}
\label{sec:ARF_IEEE}
We now outline a generalization of the previous analysis to account for the MAC overhead present in real protocols. We focus on the popular IEEE 802.11b DCF standard~\cite{ieee:80211}, with which we assume the reader is familiar.
For brevity, we only discuss the ARF algorithm, since the generalized analysis of AARF follows similar lines. Due to space constraints, detailed technical derivations are deferred to~\cite{angad-thesis}.

The first source of overhead in IEEE 802.11 is the transmission time of ACK packets, denoted by $T_{ACK}$, and the required  inter-frame spacings DIFS and SIFS. For simplicity, we assume that the RTS/CTS handshake is disabled and ACK packets are not lost. Hence, for each successful DATA packet transmission, the time overhead amounts to $T_{success}= DIFS + SIFS +T_{ACK}$, while for each failed DATA packet transmission, the overhead corresponds to $T_{failure}=DIFS$.

The second major source of overhead results from the binary exponential back-off procedure. The average back-off time, after $0 \leq \gamma \leq \gamma_{max}$ consecutive failures, is given by:
\begin{equation}\label{eq:T(n)}
\overline{T}_{b}(\gamma)=\frac{2^{\gamma}CW_{min}-1}{2},
\end{equation}
where $\gamma$ is the back-off counter. The value of this counter is incremented after each transmission failure and reset after a successful transmission.
Note that  $\gamma_{max}$ represents the maximum number of retransmission attempts before the packet is dropped. We assume that $CW_{max} \geq 2^{\gamma_{max}}CW_{min}-1$.

We now model the evolution of the ARF process, inclusive of MAC overhead. We define $i_{\gamma}$ to be a state of the process in which packet transmissions are carried out at the
bit rate $R_i$. The index $\gamma$ corresponds to the value of the back-off counter \emph{at the instant} where the process enters the state. Transitions out of state $i_{\gamma}$ take place when the bit rate changes, following the rules of ARF. Hence, $s$ successive successful packet transmissions
lead to a transition into state ${(i+1)}_0$, for $i < N$. On the other hand, $f$ consecutive packet transmission failures lead to a transition into one of two possible states. If, on entering state $i_{\gamma}$, $f$ consecutive packet transmissions fail without any interim successful packet transmissions, then state $i_{\gamma}$ is exited with
a back-off counter equalling $(\gamma+f) \pmod {\gamma_{max}+1}$. This back-off
counter is carried over to the next state implying a transition to state ${(i-1)}_{(\gamma+f)\pmod {\gamma_{max}+1}}$, for $i > 1$. Otherwise,
if a successful packet transmission occurs prior to $f$ consecutive failures, the value of the back-off counter at the exit time is ${f\pmod{ \gamma_{max}+1}}$, implying a transition to state ${(i-1)}_{f\pmod {\gamma_{max}+1}}$, for $i > 1$.

When entering state $i_{\gamma}$, the time spent by the process in that state and the transition probabilities to other states are independent of history. Hence, the process is semi-Markovian. Figure~\ref{fig:ARF_IEEE-Macro-Example} shows a state diagram of this process for the case  $f=2$,  $N=3$, and $\gamma_{max}=5$.

\begin{figure}[t]
\centering
\resizebox{3.5in}{!}{\includegraphics{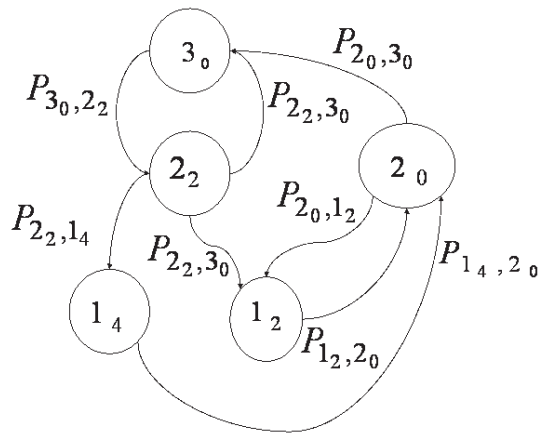}}
\caption{Embedded Markov chain modeling the behavior of ARF in IEEE
802.11 at the moments of transitions to a new state formulated for
specific values of $N=3$, $f=2$ and $\gamma_{max}=5$.}
\label{fig:ARF_IEEE-Macro-Example}
\end{figure}
\renewcommand{\baselinestretch}{\linegap}

Let $\mu_{i_{\gamma}}$ represent the expected time (inclusive of overhead) spent in state $i_{\gamma}$ and let $\pi_{i_{\gamma}}$ represent the steady-state probabilities of the embedded Markov chain associated with the semi-Markov process.  In~\cite{angad-thesis}, we detail how to compute these quantities. The analysis is similar to that carried out in the previous sections. In particular, for the case $\gamma=0$, the transition probabilities out of state $i_0$ are the same as those out of state $i$, as given by Proposition~\ref{prop2}, and the expected time spent in state $i_0$ is
\begin{eqnarray}\label{eq:mui_IEEE}
\mu_{i_0}&=&\sum_{j=1}^{f-1}\overline{Y_{i}}(-j)(\frac{\overline{\ell}}{R_{i}}+
\overline{T}_{b}(j \pmod{\gamma_{max}+1})+\alpha_{i}T_{success}+(1-\alpha_{i})T_{failure})\nonumber \\
&+&\sum_{j=0}^{s-1}\overline{Y_{i}}(j)(\frac{\overline{\ell}}{R_{i}}+
\overline{T}_{b}(0)+\alpha_{i}T_{success}+(1-\alpha_{i})T_{failure}),
\end{eqnarray}
where $\overline{Y_{i}}(j)$ represents the expected number of visits to state $\cal{S}_j^i$, which is defined the same way as in Fig.~\ref{fig:ARF-micro} (i.e., $\cal{S}_j^i$ corresponds to the state of $j$ consecutive successes or failures, depending on the sign of $j$, at bit rate $R_i$). Eq.~(\ref{eq:mui_IEEE}) accounts for both the average transmission time and overhead associated with each visit.

The average throughput of ARF, inclusive of MAC overhead, can now readily be computed. One should note, however, that in each state $i_{\gamma}$ the average time spent in actual transmissions, denoted  $\theta_{i_{\gamma}}$, is shorter than $\mu_{i_{\gamma}}$, the total average time spent in that state. The value of $\theta_{i_{\gamma}}$ is independent of $\gamma$ and identical to the average time spent in state $i$ by the semi-Markov process of ARF, exclusive of MAC overhead, that was analyzed in Section~\ref{sec:ARF} (see  Proposition~\ref{prop1} and Eq.~(\ref{eq:mui})).
Hence, the long-run proportion of time spent by ARF transmitting at rate $R_i$ is given by
\begin{equation}\label{eq:frac_back_off}
g_{i}=\frac{\sum_{\gamma=0}^{\gamma_{max}}{\pi_{i_{\gamma}}\theta_{i_{\gamma}}}}
{\sum_{i=1}^{N}\sum_{\gamma=0}^{\gamma_{max}}{\pi_{i_{\gamma}}\mu_{i_{\gamma}}}},
\end{equation}
and the final expression for the average throughput is
$
\tau=\sum_{i=1}^{N} g_{i} \alpha_{i}R_{i}$.

\section {Numerical Results}
\label{sec:Numercial Analysis}

In this section, we illustrate the utility of our analysis by
numerically comparing the throughput performance of ARF and AARF
in two different channel regimes. These results show that neither
of these algorithms consistently outperforms the other. We then
propose a new variant, called Persistent AARF (PARFF), that is
shown to achieve a good compromise between the two algorithms.
Finally, we evaluate the impact of MAC overhead on the respective performance of the algorithms.

\subsection{Performance Comparison of ARF and AARF}

We consider a wireless LAN supporting $N=2$ bit rates, with $R_1=
1$ Mbps and $R_2=2$ Mbps. The parameters of the algorithms are set
as follows: $s=10$, $f=2$, and, for AARF, $\beta_{max}=3$. We
compare the throughput performance of ARF and AARF under two
channel regimes of practical interest.  In the first regime, the
probability of a successful packet transmission at bit rate $R_1$
is much higher than at bit rate $R_2$, i.e., we fix $\alpha_2=0.2$
and evaluate the throughput of ARF and AARF for values of
$\alpha_1$ ranging from 0.7 to 1. In the second regime, the
probability of a successful packet transmission at bit rate $R_1$
is only slightly higher than at bit rate $R_2$, that is, we fix
$\alpha_2=0.7$ and vary $\alpha_1$ from 0.7 to 1 (note that
$\alpha_1$ should always exceed $\alpha_2$).

Figure~\ref{fig:0.2} depicts results for the first regime. We
observe that AARF outperforms ARF and that the difference between
the performance increases with $\alpha_1$. The cause of the
discrepancy is that ARF attempts too often to switch to the
failure prone $R_2$ bit rate, which results in throughput
degradation. On the other hand, AARF spends a lot more time in the
optimal $R_1$ bit rate. This result is consistent with
experimental and simulation results reported
in~\cite{Bicket,AARF}.

Figure~\ref{fig:0.7} shows results for the second regime and
illustrates conditions under which ARF outperforms AARF.
The throughput performance of AARF suffers in this region as it
tends to spend too much time in the under performing $R_1$
bit rate, whereas ARF tries to switch to the optimal $R_2$ much
more frequently. This is an insightful result as it indicates the
need to optimize AARF under channel regimes where the probability
of successful packet transmission is high at both the lower and
higher bit rates.

\begin{figure}[t]
\centering \resizebox{3.5in}{!}{\includegraphics{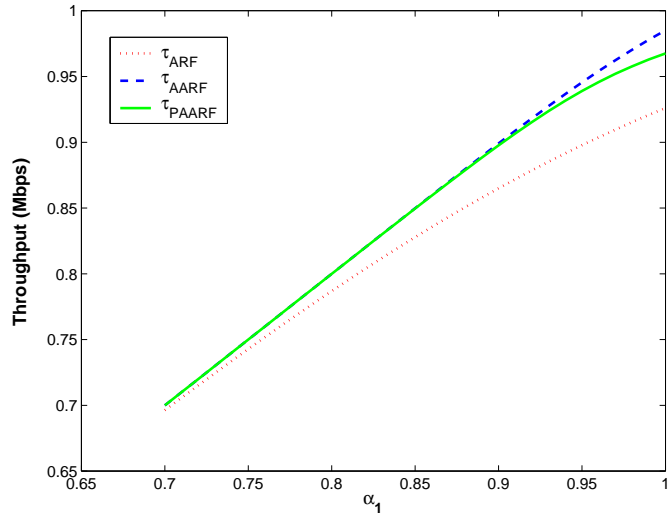}}
\caption{$R_{1}=1$ Mbps, $R_{2}=2$ Mbps, $\alpha_{2}=0.2$}
\label{fig:0.2}
\end{figure}\renewcommand{\baselinestretch}{\linegap}

\begin{figure}[t]
\centering \resizebox{3.5in}{!}{\includegraphics{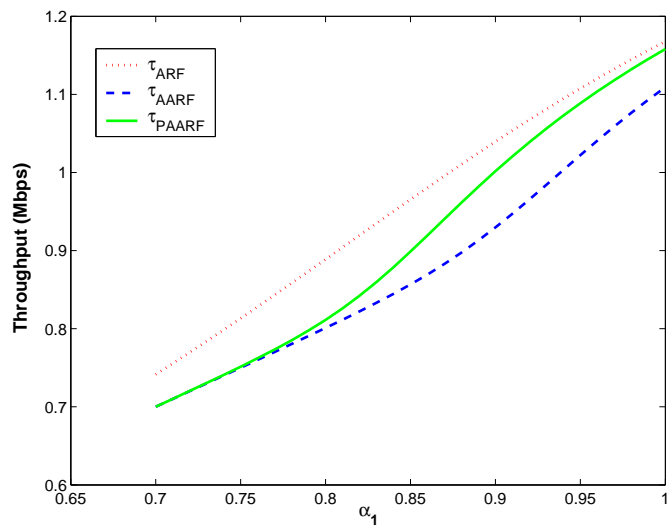}}
\caption{$R_{1}=1$ Mbps, $R_{2}=2$ Mbps, $\alpha_{2}=0.7$}
\label{fig:0.7}
\end{figure}\renewcommand{\baselinestretch}{\linegap}

\subsection{Persistent AARF}

The main cause of the relatively poor performance of AARF in the
second regime is that it is not persistent enough in probing the
higher-bit rates. Thus, we propose a simple variation of AARF,
called Persistent AARF. PAARF is identical to AARF, except that,
when entering a probe state $i^{+1}_{\beta}$, it transmits two
probe packets at the next higher bit rate instead of just one. If
anyone of these two probe packets is successfully transmitted,
then PAARF switches to the next higher bit rate, i.e., to state
$(i+1)_{0}$.

One of the main advantages of our analytical approach is to allow
evaluating the performance of such new variants without having to
run lengthy simulations. In particular, the analysis of PAARF turns out to be
almost the same as that of AARF. The only difference is the mean
time spent in probe states which is now
\begin{equation}
\mu_{i^{+1}_{\beta}}=\frac{(2-\alpha_{i+1})\overline{\ell}}{R_{i+1}},
\end{equation}
and the transition probabilities out of  the probe states which
become
\begin{equation}
p_{(i^{+1}_{\beta},(i+1)_0)}  =
2\alpha_{i+1}-(\alpha_{i+1})^2; \qquad
p_{(i^{+1}_{\beta},i_{\beta+1})}  =  1-
2\alpha_{i+1}+(\alpha_{i+1})^2,
\end{equation}
and for the case $\beta=\beta_{max}$,
\begin{equation}
p_{(i^{+1}_{\beta_{max}},(i+1)_0)} =
2\alpha_{i+1}-(\alpha_{i+1})^2; \qquad
p_{(i^{+1}_{\beta_{max}},i_{\beta_{max}})} =  1-
2\alpha_{i+1}+(\alpha_{i+1})^2.
\end{equation}

Figures~\ref{fig:0.2} and~\ref{fig:0.7} show the performance of
PAARF for the two channel regimes in consideration. As one can
see, PAARF generally performs close to the best algorithm in each
case.  One exception is when $\alpha_1$ is very close to
$\alpha_2$, in the second regime. In that case, PAARF performs
only marginally better than AARF. However, we conjecture that the
likelihood of this scenario is relatively low because if the
packet success probability at rate $R_2$ is quite high (e.g.,
0.7), then the packet success probability at rate $R_1$ is likely
to be close to 1.

\subsection{Impact of MAC Overhead and Transmission Rates}
\label{sec:impact}
We next evaluate the effects of MAC overhead on the performance of the various algorithms. We present numerical results based on the analysis of Section~\ref{sec:ARF_IEEE}. The protocol
parameters are set according to the specifications of the IEEE 802.11 standard, that is, $\gamma_{max}=5$, $DIFS=50$~$\mu$s, $SIFS=10$~$\mu$s, $T_{ACK}=112$~$\mu$s, and~$CW_{min}=32$~\cite{ieee:80211}.
The other algorithm, operational and channel parameters remain the same as in the previous sections.

Figures~\ref{fig:ARFAARFPAARF212} and~\ref{fig:ARFAARFPAARF712}  show the performance of ARF, AARF
and PAARF, respectively for $\alpha_{2}=0.2$ and $\alpha_{2}=0.7$. The bit rates are set to $R_{1}=1$ Mbps and $R_{2}=2$ Mbps. The results show that MAC overhead reduces the achieved throughput of all the three
algorithms. However their underlying behavior remains the same as previously.

We next illustrate the effect of the bit rates on the
behavior of the three algorithms. For this, we again analyze
their delivered throughput under the previously mentioned channel
regimes, but this time with bit rates $R_{1}=5.5$ Mbps and $R_{2}=11$ Mbps.

Figure~\ref{fig:ARFAARFPAARF2511} show the performance of ARF, AARF and PAARF for $\alpha_{2}=0.2$. We can see that the relative performance of ARF with respect to the other algorithms is considerably worse at higher rates. This can be attributed to the fact that MAC overhead time does not scale
with the transmission rate. Hence, at high data rate, the impact of MAC overhead due to a failed transmission gets magnified. Consequently, ARF with its higher tendency to switch to the
poor performing rate $R_{2}$ experiences a relatively larger performance
degradation than the other algorithms.

Figure~\ref{fig:ARFAARFPAARF7511} shows the performance of the three
algorithms in the second regime, i.e., $\alpha_{2}=0.7$. The nature of the
results in this region is similar to those shown in
Fig.~\ref{fig:0.7}, though the superiority of ARF is not as pronounced. In fact, as $\alpha_{1}$ tends to 1, PAARF starts outperforming ARF.
Overall, these results show that PAARF keeps providing a good trade-off between ARF and~AARF, even when taking MAC overhead into consideration.

\begin{figure}[t]
\centering
\resizebox{4in}{!}{\includegraphics{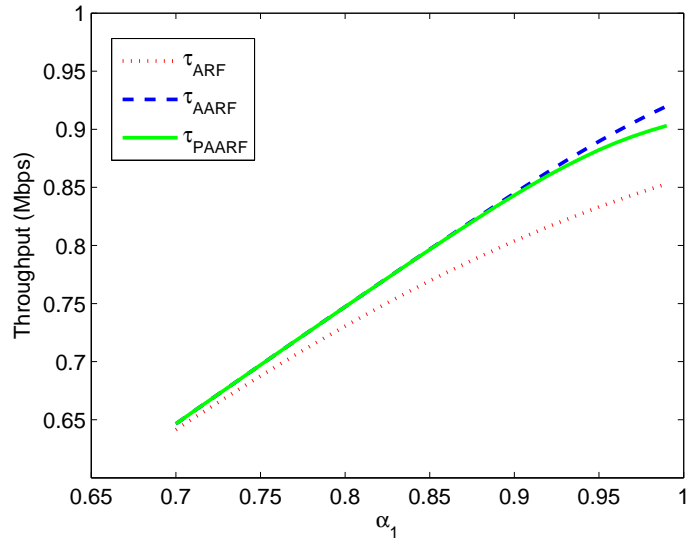}}
\caption{$R_{1}=1$ Mbps, $R_{2}=2$ Mbps, $\alpha_{2}=0.2$ with MAC
overhead} \label{fig:ARFAARFPAARF212}
\end{figure}
\renewcommand{\baselinestretch}{\linegap}

\begin{figure}[t]
\centering
\resizebox{4in}{!}{\includegraphics{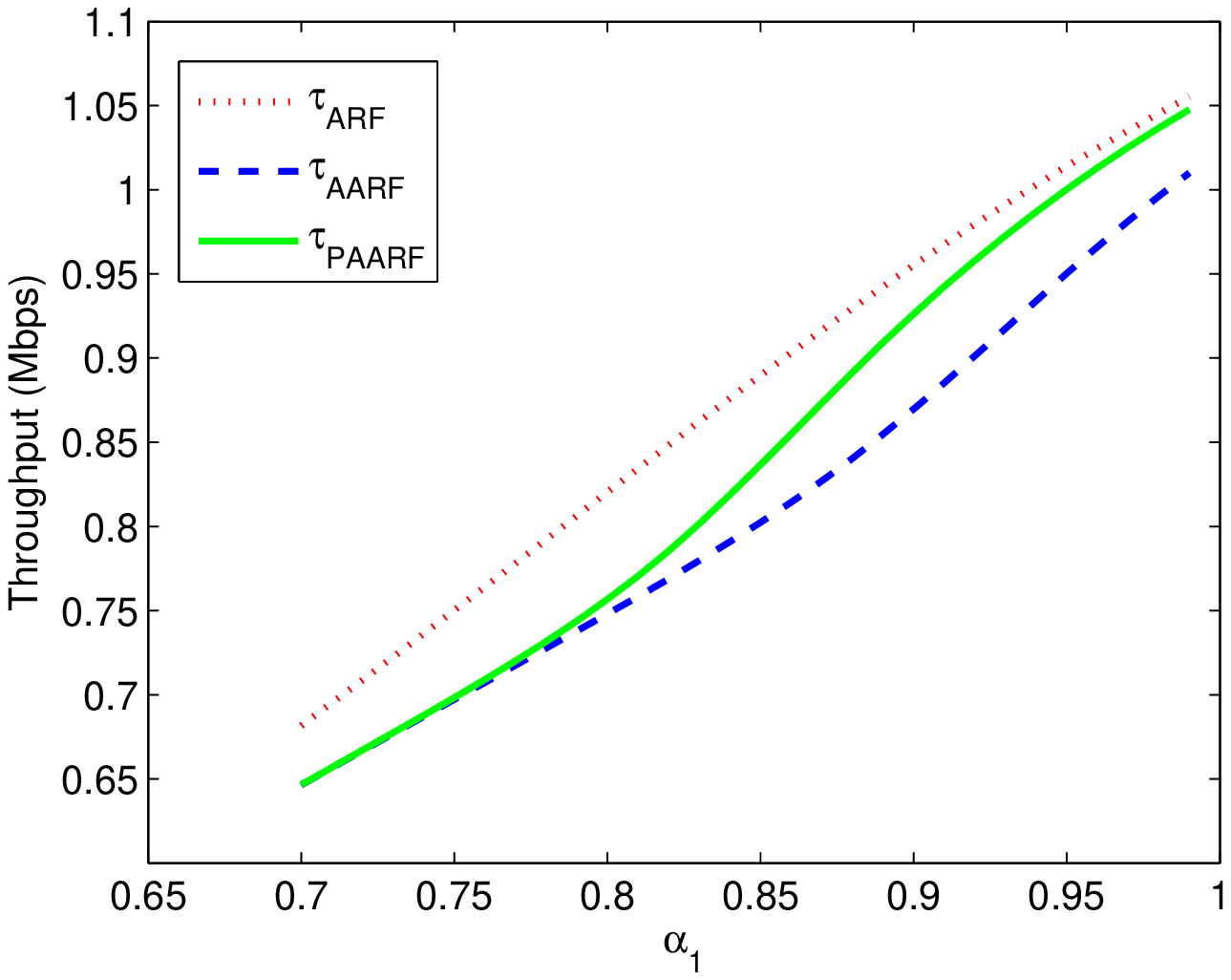}}
\caption{$R_{1}=1$ Mbps, $R_{2}=2$ Mbps, $\alpha_{2}=0.7$ with MAC
overhead} \label{fig:ARFAARFPAARF712}
\end{figure}
\renewcommand{\baselinestretch}{\linegap}

\begin{figure}[t]
\centering
\resizebox{4in}{!}{\includegraphics{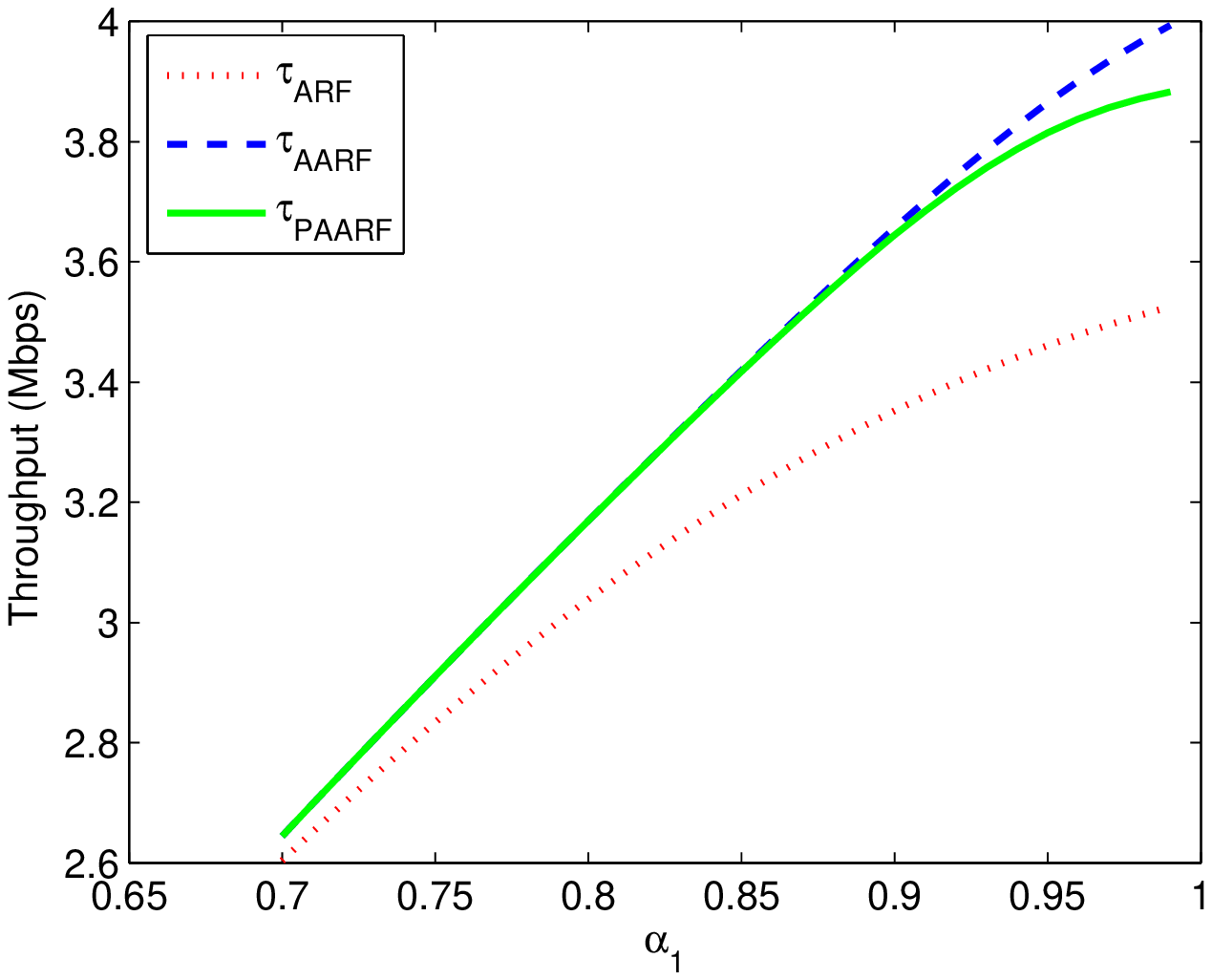}}
\caption{$R_{1}=5.5$ Mbps, $R_{2}=11$ Mbps, $\alpha_{2}=0.2$ with
MAC overhead} \label{fig:ARFAARFPAARF2511}
\end{figure}
\renewcommand{\baselinestretch}{\linegap}

\begin{figure}[t]
\centering
\resizebox{4in}{!}{\includegraphics{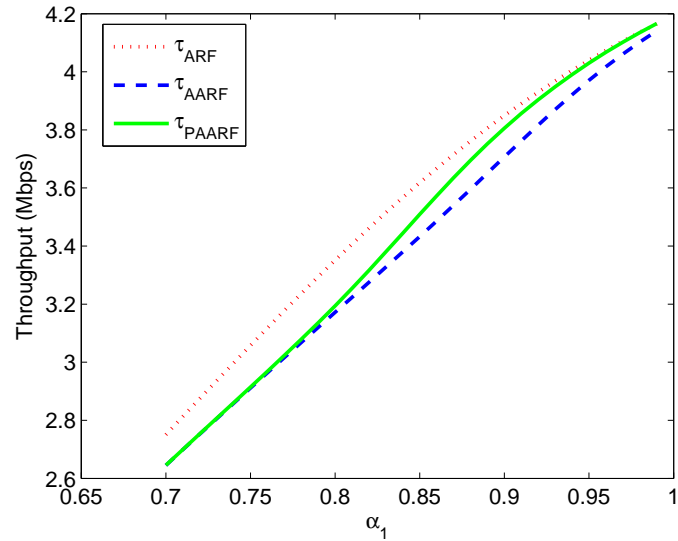}}
\caption{$R_{1}=5.5$ Mbps, $R_{2}=11$ Mbps, $\alpha_{2}=0.7$ with
MAC overhead} \label{fig:ARFAARFPAARF7511}
\end{figure}
\renewcommand{\baselinestretch}{\linegap}

\section {Conclusions}
\label{sec:Conclusions}

In this paper, we proposed a novel semi-Markovian framework
to analyze the performance of two of the most widely implemented
rate adaptation algorithms in wireless LANs, namely ARF and AARF.
Given our modeling assumptions, the analysis is exact and provides
closed form expressions for the achievable throughput of ARF and
AARF.
A particularly interesting finding was that
the multi-dimensional embedded Markov chain associated with the
semi-Markov process of AARF collapses into a simple one
dimensional birth-death process.

We used the analytical expressions to numerically compare the throughput
performance of ARF and AARF in two channel regimes for a wireless
LAN operating at two different bit rates. We found that none of the algorithms consistently prevails over the other.
Based on this insight, we devised a new variant to AARF, called
Persistent AARF (PAARF), whereby two probe packets (instead of
just one) are transmitted each time the algorithm enters one of
the probe states. We were able to analyze PAARF much the same way
as AARF and our numerical results showed that this simple
modification can significantly improve the performance of AARF in
the regime where it does not perform well, while maintaining
almost the same performance in the regime where it does perform
well.

Next we analyzed the impact of the IEEE 802.11b MAC overhead on the
performance of ARF, AARF and PAARF by numerically comparing their
throughput performance based on the analysis presented in
section~\ref{sec:ARF_IEEE}. The results revealed that, at low bit rates,
MAC overhead does not alter the basic behavior of the three
algorithms and contributes only in reducing the achieved throughput.
However, MAC overhead becomes increasingly significant as the
bit rate increases. This is because the exponentially
increasing back-off time accompanying failed packet transmissions
becomes particularly significant at high bit rates. This
phenomenon translates into a performance edge for AARF and PAARF, as
the bit rates increase.

The analytical framework developed in this paper provides the
basis for many other interesting types of optimizations. For
instance, an important issue is how to optimally set the
operational parameters of ARF and AARF.  Another important area
for future work is to numerically evaluate the performance of ARF,
AARF, and PAARF for more than two bit rates. Overall, this work marks a first step in modeling rate adaptation in wireless LANs and shows promise for analytically evaluating other open-loop rate adaptation algorithms, especially
those based on ARF.

\bibliographystyle{IEEEtran}
\bibliography{SeconPaper}

\newpage

\end{document}